\documentclass[11pt,a4paper,twoside]{amsart}
\usepackage[english]{babel}
\usepackage{enumerate,colonequals,expdlist}
\usepackage{amssymb,amsmath,amsfonts,amsthm}
\usepackage{color,graphicx}
\usepackage{a4}
\usepackage{tikz}
\usepackage{colonequals}
\usepackage{pgf}
\usepackage{enumerate}
\usepackage{exscale}
\usepackage{amssymb,amsfonts,latexsym,mathabx}
\usepackage{ifthen}

\definecolor{darkgray}{rgb}{0.5,0.5,0.5}

\usetikzlibrary{arrows,positioning}
\usetikzlibrary{patterns}
\usetikzlibrary{automata,er}
\usetikzlibrary{positioning,decorations.pathreplacing,shapes.multipart}
\usetikzlibrary{calc}

\renewcommand{\epsilon}{\varepsilon}

\newcommand{\RR}{\mathbb{R}}
\newcommand{\CC}{\mathbb{C}}
\newcommand{\NN}{\mathbb{N}}
\newcommand{\ZZ}{\mathbb{Z}}
\newcommand{\PP}{\mathbb{P}}
\newcommand{\EE}{\mathbb{E}}

\DeclareMathOperator{\Tr}{Tr}
\newtheorem{theorem}{Theorem}[section]
\newtheorem{lemma}[theorem]{Lemma}
\newtheorem{proposition}[theorem]{Proposition}
\newtheorem{corollary}[theorem]{Corollary}
\theoremstyle{definition}
\newtheorem{definition}[theorem]{Definition}
\newtheorem*{hypothesis}{Hypothesis}
\theoremstyle{remark}
\newtheorem{remark}[theorem]{Remark}
 \newcommand{\hm}[1]{\leavevmode{\marginpar{\tiny%
 $ \hbox to 0mm{\hspace*{-0.5mm} $ \leftarrow $ \hss}%
 \vcenter{\vrule depth 0.1mm height 0.1mm width \the\marginparwidth}%
 \hbox to
 0mm{\hss $ \rightarrow $ \hspace*{-0.5mm}} $ \\\relax\raggedright #1}}}

\begin{document}

\title[Wegner estimate for the random breather model]
{Conditional Wegner estimate for the standard random breather potential}

\author[M.~T\"aufer]{Matthias T\"aufer}
\author[I.~Veseli\'c]{Ivan Veseli\'c}
\address{ Fakult\"at f\"ur Mathematik,\,  TU-Chemnitz, Germany  }
\email{matthias.taeufer@mathematik.tu-chemnitz.de}
\urladdr{www.tu-chemnitz.de/mathematik/stochastik}

\thanks{
 \copyright 2014 by the authors. Faithful reproduction of this article, is permitted for non-commercial purposes.
{\today, \jobname.tex}}

\thanks{
This work has been partially supported by the DFG under grant \emph{Eindeutige Fortsetzungsprinzipien und Gleichverteilungseigenschaften von Eigenfunktionen}.
It has also profited from interactions with Francisco Hoecker-Escuti, Ivica Naki\'{c}, Martin Tautenhahn and Christoph Schumacher.
Part of these interactions have been supported by the binational German-Croatian 
DAAD project \emph{Scale-uniform controllability of partial differential equations}.}


\begin{abstract}
We prove a conditional Wegner estimate for Schr\"odinger operators with random potentials of breather type.
More precisely, we reduce the proof of the Wegner estimate to a scale free unique continuation principle.
The relevance of such unique continuation principles has been emphasized in previous papers, in particular in recent years. 

We consider the \emph{standard} breather model, meaning that the single site potential is the characteristic function of a ball or a cube.
While our methods work for a substantially larger class of random breather potentials, we discuss in this particular paper only the standard model in order to make the arguments and ideas easily accessible.
\end{abstract}

\maketitle

\section{Introduction}
A Wegner estimate is an upper bound on the expected number of eigenvalues in a prescribed energy interval of a finite box Hamiltonian.
The expectation here refers to the potential which is random. 
Wegner estimates have been derived for Hamiltonians living on $\ZZ^d$ or $\RR^d$, more precisely on bounded subsets of rectangular shape of these spaces. 
In the present note we do not put emphasis on presenting the history of various variants of the Wegner estimate, but rather refer to the monograph \cite{Veselic-07b} and recent papers  \cite{RojasMolinaV-13}, and \cite{Klein-13}.
Here we will be only interested in models on continuum space $\RR^d$. 
The most studied example in this situation is the so called alloy-type potential, sometimes also called continuum Anderson model. 
A particular feature of this model is that randomness enters the model via a countable number of random variables, and these r.v. influence the potential in a linear way.
In the model we study here this dependence is no longer linear, but becomes non-linear.
What remains, is the monotone dependence of the potential on the r.v.
The topic of the present note is to explain, how to effectively use this monotonicity.
This only works if it is possible to cast the monotonicity in a quantitative form.
With this respect we consider the study of the random breather model as paradigmatic for a better understanding of random Schr\"odinger operators with non-linear randomness.
Note that in the case that one aims to establish a Wegner estimate only in an energy region near the bottom of the spectrum, it is possible to overcome the monotonicity assumption.
A paradigmatic example is the random displacement model, cf.~\cite{KloppLNS-12}. 
However let us clarify that any Wegner estimate proven so far uses at least one of the following properties
\begin{enumerate}[(a)]
\item
The space dimension is one.
\item
The random potential has a quantitative monotonicity property.
\item
One restricts attention to an energy region near the spectral bottom (or some other spectral edge, however this requires again extra assumptions, like small disorder).
\end{enumerate}
A strategy to prove Wegner estimates without relying on any of the above properties seems to require a truly probabilistic approach, rather than an analyitic one, like averaging just over a finite number of random variables. 
\medskip

To our best knowledge random breather potentials have been first considered in the mathematical physics literature in the work \cite{CombesHM-96}.
A Wegner estimate for the random breather model was derived in \cite{CombesHN-01} and a Lifschitz tail bound, yielding localization, in \cite{KirschV-10}.
However, all these papers have a condition on the gradient of the single site potential, arising from  a linearization. 
Unfortunately, this excludes the most elementary single site potential, namely the characteristic function of a ball or a cube.
A simple situation where this was overcome is treated in \cite{Veselic-07} where a Lifschitz bound was proven. 
Actually this proof extends to very general breather models, as will be explained in \cite{SchumacherV}.
\medskip

The aim of this note is not to cover the  most general types of random breather potentials, but to concentrate on the simplest case and provide full proofs and calculations
accessible to 
non-specialists. For this purpose we spell out explicitly the theorems which we infer, rather than 
just giving references to earlier papers.
Our proof starts from a scale free unique continuation principle (SFUCP). 
Such an estimate has recently been proven, for the case that the magnetic vector potential vanishes.
It is announced in \cite{NakicTTV-14} and full proofs will be presented in \cite{NakicTTV-15}, which rely on Carleman estimates, interpolation inequalities and related PDE techniques.
The proofs in this note are more probabilistic in nature.

%%%%%%%%
%%%%%%%%
\subsection{Wegner estimate for the random breather model}
We prove a Wegner estimate, Theorem \ref{Wegner Breather}, for the Random Breather model.
In the following, $\Lambda_s(x) := x + (-s/2, s/2)^d$ is the $d$-dimensional open cube of side length $s>0$, centered at the point $x \in \RR^d$. $B_r(x)$ denotes the open ball of radius $r \geq 0$ around $x \in \RR^d$. If $x = 0$, we omit the $x$ and write $\Lambda_s$ or $B_r$. 
If the side length $s$ is fixed, we simply write $\Lambda$.
Let $0 \leq \omega_{-} < \omega_{+} < 1/2$ and let $\mu$ be a probability measure on $\RR$ with bounded density $\nu_\mu$ and support in $[\omega_{-}, \omega_{+}]$. 
We define the probability space 
\begin{equation*}
\left( \Omega, \mathcal{A}, \PP \right) = \left( \times_{i \in \ZZ^d} \RR, \otimes_{i \in \ZZ^d} \mathcal{B}(\RR), \otimes_{i \in \ZZ^d} \mu \right).
\end{equation*}
Here, $\mathcal{B}$ is the Borel $\sigma$-algebra. For $\omega \in \Omega$ and $j \in \ZZ^d$ we denote the projection onto the $j$-th coordinate of $\Omega$ by $\omega_j$. 
The $\{ \omega_j \}_{j \in \ZZ^d}$ form a process of $[\omega_{-}, \omega_{+}]$-valued independent and identically distributed random variables on $\ZZ^d$. For $i,j \in \ZZ^d$ and $\delta \in \RR$ we define $\omega + \delta$ and $\omega + \delta e_i \in \Omega$ by
\begin{align*}
\left( \omega + \delta \right)_j &:= \omega_j + \delta \mbox{ for all } j \in \ZZ^d\\
\left( \omega + \delta e_i \right)_j &:= 
\begin{cases}
\omega_j + \delta &\mbox{ if } j=i\\
\omega_j &\mbox{ if } j \neq i.
\end{cases}
\end{align*}
\begin{definition}
We consider a  nonnegative and self-adjoint magnetic Schr\"odinger operator with vector potential $A \in L^2_{\mathrm{loc}}(\RR^d, \RR^d)$
\begin{equation*}
H_A := ( - i \nabla - A)^2 
\end{equation*}
on $L^2(\RR^d)$. 
We define a random potential by   
\begin{align}
& V_\omega(x) := \sum_{j \in \ZZ^d} \chi_{B_{\omega_j}}(x-j)\quad \mbox{or} \label{Breather model balls} \\
& V_\omega(x) := \sum_{j \in \ZZ^d} \chi_{\Lambda_{2 \omega_j}}(x-j) \label{Breather model cubes}
\end{align}
and define the corresponding random operator $H_\omega := H_A + V_\omega$.
In both cases we call  $H_\omega $ the \emph{standard random breather model}.
For $L \in \NN_{\mathrm{odd}} = \{ 1,3,... \}$ we define the restriction $H_{\omega, L}$ of $H_\omega$ to $\Lambda_L$ with Dirichlet boundary conditions and the restriction of the potential $V_{\omega,L} : \Lambda_L \rightarrow \RR$.
\end{definition}
\begin{definition}
 Let $\delta \in (0, 1/2)$. A sequence $\{ x_j \}_{j \in \ZZ^d}$ is called \emph{$\delta$-equidistributed} if for every $j \in \ZZ^d$ we have $B_{\delta}(x_j) \subset \Lambda_1(j)$.
 For such a sequence and $L \in \NN_{\mathrm{odd}}$, we define $W_{\delta,L} : \Lambda_L \to \CC$ as the characteristic function of $\bigcup_{j \in \ZZ^d } B_\delta(x_j) \cap \Lambda_L$.
\end{definition}
We formulate a scale free quantitative unique continuation property, which the random operator $H_{\omega}$
may or may not have.
Denote by $\chi_I(H_{\omega, L})$ the spectral projector of $H_{\omega, L}$ onto an interval $I$.
\begin{hypothesis}[SFUCP]\label{hypothesis UCP}
\label{eq:ball_in_box} 
Given $b \in \RR$ there  are  $M\geq 1$, $\kappa \leq 1$  such that for all $L \in \NN_{\mathrm{odd}}$, almost all $\omega \in \Omega$ and for all $0 < \delta < 1/2 - \omega_{+}$ we have 
\begin{equation}\label{eq:abstractUCP}
\chi_{(- \infty , b]}(H_{\omega, L}) \ W_{\delta,L} \ \chi_{(- \infty , b]}(H_{\omega, L}) 
\geq  \kappa \, \delta^{1/M} \chi_{(- \infty , b]}(H_{\omega, L}) .
\end{equation}
Inequality \eqref{eq:abstractUCP} is understood in the sense of quadratic forms.
\end{hypothesis}
\medskip
\begin{theorem}[Wegner estimate for the Random Breather model]\label{Wegner Breather}
Assume that the standard random breather model satisfies hypothesis SFUCP. Let $b \in \RR$. 
Then there are constants  $C=C(d,b) \in (0,\infty)$, $\epsilon_{\max}=\epsilon_{\max}(\kappa, M, \omega_{+}) \in (0,\infty)$
such that for all $L \in \NN_{\mathrm{odd}}$, all $0 < \epsilon \leq \epsilon_{\max}$, $E \in \RR$ with
$[E-\epsilon, E+\epsilon] \subseteq (- \infty , b-1]$
we have 
\begin{equation}\label{eq: Wegner Breather}
\EE \left[ \mathrm{Tr} \left[ \chi_{[E- \epsilon, E + \epsilon]}(H_{\omega,L}) \right] \right]
\leq 
C  
\sqrt[M]{4 \kappa^{-1}} 
\|\nu_\mu\|_\infty
\epsilon^{1/M} 
\left\lvert\ln \epsilon \right\rvert^d L^d.
\end{equation}
$\epsilon_{\max}$ can be chosen as
\begin{equation*}
\epsilon_{\max} =
\frac{\kappa}{4} 
\left( \frac{1/2 - \omega_{+}}{2} \right)^M.
\end{equation*}
\end{theorem}
\begin{remark}
\begin{enumerate}[(i)]
 \item We can give an explicit bound on the constant $C$, namely 
\begin{equation*}
C \leq 2 \cdot 32^d (2 e^b \cdot (d+1)! + 2^d).
\end{equation*}
 \item The $\lvert \ln \epsilon \rvert^d$ term can be hidden by choosing a slightly smaller $\tilde{M} < M$ and a different constant $\tilde{C} > 0$. Hiding also the factor $\sqrt[M]{4\kappa^{-1}} \|\nu_\mu\|_\infty$ in this constant, we obtain a bound of the form
\begin{equation*}
\EE \left[ \mathrm{Tr} \left[ \chi_{[E- \epsilon, E + \epsilon]}(H_{\omega,L}) \right] \right]
\leq
\tilde{C} \cdot \epsilon^{1/\tilde{M}} L^d.
\end{equation*}
This implies local Hoelder continutity of the integrated density of states with the exponent $1/\tilde M$ as in \cite{HundertmarkKNSV-06}
 \item If SFUCP only holds for $L \in \mathcal{I} \subset \NN_{\mathrm{odd}}$ then Theorem \ref{Wegner Breather} holds for $L \in \mathcal{I}$.
\end{enumerate}
\end{remark}

%%%%%%%%
%%%%%%%%
\section{Estimates on the spectral shift function}
\label{s:SSF}
We infer  here the estimates on the singular values of semingroup differences and the 
spectral shift function of two Schr\"odinger operators differing by a compactly supported potential
obtained in  \cite{HundertmarkKNSV-06}.
We reduce the estimates to the particular, simple situation we are dealing with. The constants are calculated explicitly and more accurately than in \cite{HundertmarkKNSV-06}.
We will be dealing here with a pair of operators: $H_0 = H_A + V_0$ with $V_0$ non-negative and bounded, and 
$H_1 := H_0 + V$, with $V$ non-negative, bounded, and supported in a cube $\Lambda_{1}$ of sidelength one.
Let $\Lambda = \Lambda_L$ be a cube of sidelength $L \in \NN_{\mathrm{odd}}$ and call $H_0^\Lambda$ and $H_1^\Lambda$ the restrictions of $H_0$ and $H_1$ onto $\Lambda$ with Dirichlet boundary conditions. 
Let
\begin{equation*}
V_{\mathrm{eff}}^\Lambda := e^{-H_1^\Lambda} - e^{-H_0^\Lambda}.
\end{equation*}
This is a compact operator and we will enumerate its singular values decreasingly by $\mu_1 \geq \mu_2 \geq ...$. 
Then we have the following Theorem:
\begin{theorem}[Theorem 1 from \cite{HundertmarkKNSV-06}]\label{thm:singular-values}
For $n >  N_0 := 4^d$, the singular values of $V_{\mathrm{eff}}^\Lambda$ obey
\begin{equation*}
\mu_n(V_{\mathrm{eff}}^\Lambda) \leq (\sqrt[4]{2d}+1) \exp\left(- \frac{ n^{1/d}}{16}\right)
\end{equation*}
\end{theorem}
We start the proof with a lemma
\begin{lemma}\label{Lemma SSF}
Let $H = H_0$ or $H_1$ be as above and let $H^\mathcal{U}$ be the Dirichlet restriction of $H$ 
to an open set $\mathcal{U}$ with finite volume $|\mathcal{U}|$. 
Then the $n^{th}$ eigenvalue $E_n$ of $H^\mathcal{U}$ satisfies
\begin{equation}\label{equation Lemma SSF}
E_n \geq \frac{2 \pi d}{e} \left( \frac{n}{|\mathcal{U}|} \right)^{2/d}.
\end{equation}
\end{lemma}
\begin{proof}
We have $H^U \geq  H_A^\mathcal{U}$ where $H_A^\mathcal{U}$ is the Dirichlet restriction of $H_A$ to $\mathcal{U}$.
Hence
\begin{equation*}
\mathrm{Tr} \left( e^{-2 t H^\mathcal{U}} \right) \leq \mathrm{Tr} \left( e^{-2 t H_A^\mathcal{U}} \right) = \left\| e^{-t H_A^\mathcal{U}} \right\|^2_{\mathrm{HS}}
= \int \int_{\mathcal{U} \times \mathcal{U}} | e^{-t H_A^\mathcal{U}}(x,y)|^2 \mathrm{d}x\ \mathrm{d}y
\end{equation*}
where $\| \cdot \|_{\mathrm{HS}}$ denotes the Hilbert-Schmidt norm. 
The diamagnetic inequality for the 
Schr\"odinger semigroup, see \cite[Remark 1.2.iii]{HundertmarkS-04} implies the pointwise bound
\begin{align*}
 |e^{-t H_A^{\mathcal{U}}}(x,y)| \leq e^{t \Delta^\mathcal{U}}(x,y) .
\end{align*}
Using that the kernel of the Dirichlet semigroup is bounded by the free kernel, that is
\begin{align*}
&\int_{\mathcal{U}} \lvert e^{-t \Delta^\mathcal{U}}(x,y) \rvert^2 \mathrm{d}y
\leq \int_{\mathcal{U}} \lvert e^{-t \Delta}(x,y)\rvert^2 \mathrm{d}y 
= \int_{\mathcal{U}} (4 \pi t)^{-d} \exp \left( - \frac{ |x-y|^2}{2t} \right) \mathrm{d}y\\
&\quad \leq  (8 \pi t)^{-d/2}  \int_{\RR} (2 \pi t)^{-d/2} \exp \left( - \frac{|x-y|^2}{2t} \right) \mathrm{d}y 
=  (8 \pi t)^{-d/2},
\end{align*}
see \cite{BroderixHL-00} for the boundedness by the free kernel and \cite{Evans-98} Theorem 2.3.1. for the free kernel, 
we estimate
\begin{equation*}
\| e^{-t H_A^U} \|^2_{\mathrm{HS}} 
\leq |\mathcal{U}| ( 8 \pi t )^{-d/2}.
\end{equation*}
Thus
\begin{equation*}
\mathrm{Tr} (e^{-2tH^\mathcal{U}}) \leq | \mathcal{U}| ( 8 \pi t )^{-d/2}.
\end{equation*}
Denote by $\mathcal{N}^\mathcal{U}(E)$ the number of eigenvalues of $H^\mathcal{U}$ smaller or equal to $E$. 
We estimate using the above bound
\begin{align*}
\mathcal{N}^\mathcal{U}(E)
& = \int_{- \infty}^E \mathrm{d} \mathcal{N}^\mathcal{U}(s) 
\leq e^{2tE} \int_{- \infty}^\infty e^{-2ts} \mathrm{d} \mathcal{N}^\mathcal{U}(s) \\
& = e^{2tE} \mathrm{Tr}(e^{-2tH^\mathcal{U}} )
 \leq | \mathcal{U}| \cdot ( 8 \pi t)^{-d/2} e^{2tE} = | \mathcal{U} | \left( \frac{e E}{2 \pi d} \right)^{d/2},
\end{align*}
where in the last equality $t := \frac{d}{4E}$ has been chosen. Since $n \leq \mathcal{N}^\mathcal{U}(E_n)$, this implies \eqref{equation Lemma SSF}.
\end{proof}
\begin{proof}[Proof of Theorem \ref{thm:singular-values}]
The $n^{th}$ singular value will be estimated by Dirichlet decoupling at a 
scale $R_n$ which monotonously depends on $n$.
Recall that $\mathrm{supp}(V) \subseteq \Lambda_{1}(x)$ for some $x \in \RR^d$. Choose a large 
$R = R_n > 2$ to be specified later. Call $H_j^R$ ($j = 0$ or $1$) the Dirichlet restriction of $H_j$ onto $\Lambda_{2R} = (-R; R)^{d}$ and let
\begin{equation*}
A_R := e^{-H_1^R} - e^{-H_0^R} \quad \mbox{and} \quad D_R := V_{\mathrm{eff}}^\Lambda - A_R.
\end{equation*}
We apply Lemma \ref{Lemma SSF} and find
\begin{equation*}
\mu_n(e^{-H_j^R}) \leq  \exp(-\frac{  \pi d}{2 e } n^{2/d} R^{-2} ) \leq \exp \left( - \frac{1}{16} n^{2/d} R^{-2} \right)
\end{equation*}
for $j = 1,2$.
Since $A_R$ is the difference of two nonnegative operators, its singular values obey the same bound:
\begin{equation*}
\mu_n(A_R) \leq \mu_n(e^{-H_2^R}) \leq \exp \left( - \frac{1}{16} n^{2/d} R^{-2} \right).
\end{equation*}
If the operator $D_R$ is bounded, then $\mu_n(V_{\mathrm{eff}}^\Lambda) \leq \mu_n(A_R) + \| D_R \|$. Our goal is therefore to estimate the norm of $D_R$ by using the Feynman-Kac formula for Schr\"odinger semigroups with Dirichlet boundary conditions, see \cite{BroderixHL-00} and \cite{Simon-79c}.
Let $\EE_x$ and $\PP_x$ denote the expectation and probability for a Brownian motion $b_t$, starting at $x$. Let $\tau_\Lambda := \inf \{ t > 0| b_t \notin \Lambda \}$ be the exit time from $\Lambda$ and $\tau_R := \inf \{ t > 0| b_t \notin \Lambda_{2R} \}$ be the exit time from $\Lambda_{2R}$. Then
\begin{align*}
& (D_R f) (x) = \\ 
&= \EE_x \left[ e^{-i S_A(b)} \left( e^{- \int_0^1(V_0 + V)(b_s) \mathrm{d}s} - e^{- \int_0^1 V_0(b_s) \mathrm{d}s} \right) \chi_{\left[\tau_\Lambda > 1\right]}(b) \chi_{\left[\tau_R \leq 1\right]}(b) f(b_1) \right]
\end{align*}
where $S_A(b)$ is the real valued stochastic process corresponding to the purely magnetic part of the Schr\"odinger operator. 
To be precise, in order to use the Feynman-Kac formula from \cite{Simon-79c} one needs $\mathrm{div} A = 0$. Therefore we first choose the Coulomb gauge which implies $\mathrm{div} A = 0$ and then use gauge invariance as in \cite{Leinfelder-83} for $A \in L^2_{\mathrm{loc}}(\RR^d, \RR^d)$.
We take the modulus whence the magnetic part drops out and using furthermore that $\chi_{[\tau_\Lambda > 1]}(b) \leq 1$, we find
\begin{align*}
| D_R f |(x) \leq \EE_x \left[ e^{-\int_0^1 V_0(b_s) \mathrm{d}s} 
\cdot 
\lvert e^{-\int_0^1 V(b_s) \mathrm{d}s} - 1 \rvert 
\cdot
\chi_{\left[\tau_R \leq 1\right]}(b) 
\cdot
|f(b_1)| \right].
\end{align*}
Only Brownian paths which both visit $\mathrm{supp}(V)$ and leave $B_R$ within one unit of time contribute to the expectation. Thus, if $\tau_V$ is the hitting time for $\mathrm{supp} (V)$ and $\mathfrak{B} = \{ \tau_R \leq 1, \tau_V \leq 1 \}$, then
\begin{equation*}
|D_R f|(x) \leq \EE_x \left[ e^{-\int_0^1 V_0(b_s) \mathrm{d}s} 
\cdot
\lvert e^{-\int_0^1 V(b_s) \mathrm{d}s} - 1 \rvert 
\cdot 
\chi_\mathfrak{B}(b) 
\cdot
|f(b_1)| \right].
\end{equation*}
Applying H\"older's inequality yields
\begin{align*}
|D_n f|(x) 
& \leq \left( \EE_x \left[ e^{-8 \int_0^1 V_0(b_s) \mathrm{d}s} \right] \right)^{1/8} 
\left( \EE_x \left[ \lvert e^{-\int_0^1 V(b_s) \mathrm{d}s} - 1 \rvert^8 \right] \right)^{1/8}\\
& \times \left( \EE_x \left[ \chi_\mathfrak{B} (b) \right] \right)^{1/4} \left( \EE_x \left[ | f(b_1)|^2 \right] \right)^{1/2}.
\end{align*}
Since $V, V_0 \geq 0$, the first two terms are bounded by $1$.
\par
Next, we estimate $\EE_x [ \chi_\mathfrak{B} (b) ] = \PP_x ( \mathfrak{B} )$.
Letting $b_t = (b_t^{(1)}, ..., b_t^{(n)}) \in \RR^d$ and calling $\tau_r^{(j)} := \inf\{ t > 0  | b_t^{(j)} \notin (-r,r) \}$ the exit time of the $j^\mathrm{th}$ coordinate from the interval $(-r,r)$, where we choose $r := \mathrm{dist} \{ \mathrm{supp}\ V, \Lambda_{2R}^c \} \geq 3/2$, we have
\begin{align*}
\PP_0[ \tau_r \leq 1] 
&= \PP_0 \left[ \bigcup_{j = 1}^d \left\{ \tau_r^{(j)} \leq 1 \right\} \right]
 \leq \sum_{j = 1}^d \PP_0 \left[ \tau_r^{(j)} \leq 1 \right]
 = d \cdot \PP_0[ \tau_r^{(1)} \leq 1].
\end{align*} 
The projection onto the first coordinate of $(b_t)$ is a one-dimensional Brownian motion and by the reflection principle 
\begin{equation}
\PP_0[ \tau_r^{(1)} \leq 1 ] = 2 \PP_0[ |b_1^{(1)} \geq r ] = \frac{4}{ \sqrt{2 \pi} } \int_r^\infty e^{-x^2/2} \mathrm{d}x \leq \frac{4}{ \sqrt{2 \pi}} r^{-1} e^{-r^2/2}.
\end{equation}
Recalling that $r > 3/2$, and hence $\frac{4}{\sqrt{2 \pi}} r^{-1} \leq 2$ we find 
\begin{equation*}
\PP_0 [ \tau_r \leq 1] \leq 2d \cdot e^{-r^2/2}.
\end{equation*}
Since every path in $\mathfrak{B}$ must cover the distance $r \geq 3/2$ between $\mathrm{supp}(V)$ and the complement of $\Lambda_{2R}$, we find $\PP_x [ \mathfrak{B} ] \leq d \cdot e^{-r^2/2}$. 
We assumed that $R > 2$, hence $r \geq R - 1/2 \geq R/\sqrt{2}$. Then $\PP_x [ \mathfrak{B} ] \leq d e^{-R^2/4}$ and 
\begin{equation*}
|D_R f|(x) \leq \sqrt[4]{2d}  \cdot e^{-R^2/16} \left( \EE_x \lvert f(b_1) \rvert^2 \right)^{1/2} 
= \sqrt[4]{2d}  \cdot e^{-R^2/16} \left( ( e^\Delta |f|^2)(x) \right)^{1/2}.
\end{equation*}
Using the fact, that $e^\Delta$ is an $L^1$ contraction
\begin{equation*}
\| D_n f \|_2 \leq \sqrt[4]{2d} \cdot e^{-R^2/16} \| ( e^\Delta |f|^2) \|_1^{1/2} \leq \sqrt[4]{2d} \cdot e^{-R^2/16} \| f^2 \|_1^{1/2} = \sqrt[4]{2d} \cdot e^{-R^2/16} \|f\|_2.
\end{equation*} 
To balance between the two bounds obtained for $\mu_n(A_R)$ and $\|D_R\|$, we choose $R := n^{1/2d}$ and find
\begin{align*}
\mu_n (A_R) & \leq \exp ( - 1/16 \cdot n^{1/d}),\\
\| D_n \| & \leq \sqrt[4]{2d} \cdot \exp( - 1/16 \cdot n^{1/d} ),
\end{align*}
that is 
\begin{equation*}
\mu_n(V_{\mathrm{eff}}^\Lambda) \leq  (\sqrt[4]{2d}+1) \cdot \exp{ ( - 1/16 n^{1/d} ) }  .
\end{equation*}
We assumed $R = n^{1/2d} > 2$, thus this only works for
\begin{equation*}
n > N_0 = 4^d .
\end{equation*}
\end{proof}
Let now $g \in C^\infty(\RR)$ with compactly supported derivative. 
If $g(H_1) - g(H_0)$ is trace class then there is a unique function $\xi(\lambda, H_1, H_0)$
called the Lifshitz-Krein spectral shift function, such that
\begin{equation}\label{eq: Krein identity}
\mathrm{Tr} \left[ g(H_1) - g(H_0) \right] = \int  \xi(\lambda, H_1, H_0) \mathrm{d}g(\lambda).
\end{equation}
This is referred to as Krein's trace identity. $\xi(\cdot, H_1,H_0)$ is independent of the choice of $g$.
In fact, $g$ can be chosen from a substantially larger class of functions: 
\begin{proposition}[Chapter 8.9, Theorem 1 in \cite{Yafaev-92}]
Let $H_0$, $H_1$ be positive definite and $H_1^{- \tau} - H_0^{-\tau}$ trace class for some $\tau > 0$.
Furthermore let $g$ have two locally bounded derivatives and satisfy 
\begin{equation*}
| (\lambda^{\tau + 1} g^\prime(\lambda))^\prime| \leq C \lambda^{-1 - \epsilon},\ \mbox{as}\ \lambda \to \infty
\end{equation*}
for some $C, \epsilon > 0$.
Then \eqref{eq: Krein identity} holds for $g$.
\end{proposition}
For such admissible functions, it is possible to make a change of variables
\begin{equation*}
\xi(\lambda, H_1, H_0) = \mathrm{sgn}(g^\prime)\xi(g(\lambda, g(H_1), g(H_0)),
\end{equation*}
see \cite{Yafaev-92}, Chapter 8.11. This is referred to as the invariance principle for the spectral shift function.
With our choice of $H_0$ and $H_1$, $g(\lambda) := \exp(-\lambda)$ is an admissible function, 
as can be seen via Lemma \ref{Lemma SSF}.
We also define functions $F_t : [0, \infty) \to [0, \infty)$ for $t > 0$ by
\begin{equation*}
F_t(x) := \int_0^x (\exp ( t y^{1/d}) - 1) \mathrm{d}y.
\end{equation*}
\begin{theorem}[Theorem 2 from \cite{HundertmarkKNSV-06}]\label{thm:SSF}
Let $\xi$ be the spectral shift function for the pair of operators $H_0^\Lambda$ and $H_1^\Lambda$. 
\begin{enumerate}[(i)]
  \item There is a constant $K_1$, depending on $t$ such that for small enough $t > 0$
\begin{equation*}
\int_{- \infty}^T F_t(|\xi(\lambda|) d \lambda \leq K_1 e^T < \infty.
\end{equation*}
$t$ can be chosen to be $t := 1/32$ in which case $K_1 \leq 2 \cdot 32^d \cdot (d+1)!$
\item There are constants $K_1$, $K_2$, only depending on $d$ such that for any bounded function $f$ 
with compact support within $(-\infty, b]$ we have
\begin{equation*}
\int f(\lambda) \xi(\lambda) \mathrm{d}\lambda \leq K_1 e^b + K_2 \left( \log ( 1 + \| f \|_\infty ) \right)^d \| f \|_1.
\end{equation*}
We may choose $K_1 := 2 \cdot 32^d (d+1)!$ and $K_2 := 32^d$.
\end{enumerate}
\end{theorem}
\begin{corollary}
Let $g \in C^\infty(\RR)$ with $g'$ of compact support within $(-\infty,b]$ such that 
$g(H_1^\Lambda) -g(H_0^\Lambda)$ is trace class. Then
\begin{equation}
 \Tr[g(H_1^\Lambda) -g(H_0^\Lambda)]
\leq 32^d\left [ 2(d+1)! e^b + \left( \log ( 1 + \| g' \|_\infty ) \right)^d \| g' \|_1\right]
\end{equation}
\end{corollary}
\begin{proof}[Proof of Theorem \ref{thm:SSF}]
(i) Using the invariance principle for the spectral shift function we have
\begin{align*}
\int_{- \infty}^T F_t(| \xi( \lambda, H_1^\Lambda, H_0^\Lambda)|) \mathrm{d}\lambda 
& = \int_{- \infty}^T F_t(|\xi(e^{-\lambda}, e^{-H_1^\Lambda}, e^{-H_0^\Lambda})|) \mathrm{d}\lambda\\
& \leq e^T \int_{e^{-T}}^\infty F_t(|\xi(s, e^{-H_1^\Lambda}, e^{-H_0^\Lambda})|) \mathrm{d}s.
\end{align*}
Since the difference $V_{\mathrm{eff}}^\Lambda = e^{-H_1^\Lambda} - e^{-H_0^\Lambda}$ is trace class, 
we can apply an estimate of \cite{HundertmarkS-02} and find using Theorem \ref{thm:singular-values}
\begin{align*}
\int_{- \infty}^\infty F_t(|\xi(s, e^{-H_1}, e^{-H_0})|) \mathrm{d}s 
& \leq \sum_{n = 1}^\infty \mu_n(V_{\mathrm{eff}}) (F_t(n) - F_t(x-1))\\
= \sum_{n = 1}^{N_0} \mu_n(V_{\mathrm{eff}}) \int_{n-1}^n (e^{t s^{1/d}} - 1) \mathrm{d}s 
& + \sum_{n = N_0 + 1}^ \infty \mu_n(V_{\mathrm{eff}}) \int_{n-1}^n (e^{t s^{1/d}} - 1) \mathrm{d}s \\
\leq N_0 ( e^{t N_0^{1/d}} - 1) & + (\sqrt[4]{2d}+1) \sum_{n = 1}^\infty e^{(t-1/16) n^{1/d}}.
\end{align*}
If we choose $t$ smaller than $1/16$ this will be finite. We choose $t := 1/32 $ and recall that $N_0 = 4^d$ so that we obtain
\begin{equation*}
\int_{- \infty}^\infty F_t(|\xi(s, e^{-H_1}, e^{-H_0})|) \mathrm{d}s \leq 4^d ( e^{1/8} - 1) + (\sqrt[4]{2d}+1) \sum_{n = 1}^\infty e^{-1/32 n^{1/d}}.
\end{equation*}
To estimate the second summand we use
\begin{equation*}
\sum_{n = 1}^\infty e^{-1/32 n^{1/d}} \leq \int_{0}^\infty e^{-1/32 x^{1/d}} \mathrm{d}x = d \cdot 32^d \int_0^\infty e^{-y} y^{d-1} \mathrm{d}y = d \cdot 32^d \Gamma(d) = d! \cdot 32^d
\end{equation*}
and find
\begin{align*}
\int_{- \infty}^\infty F_t(|\xi(s, e^{-H_1}, e^{-H_0})|) \mathrm{d}s \leq 4^d +  (\sqrt[4]{2d}+1) \cdot 32^d d! \leq 2 \cdot 32^d \cdot (d+1)!
\end{align*}
To see the very last inequality we have to distinguish between $d = 1$ where it can be verified directly and $d = 2,3,...$ in which case we use $\sqrt[4]{2d} + 1 \leq d + 1$.

(ii) We use Young's inequality and dualize the bound from part (i). $F_t$ is non-negative and convex with $F_t'(0) = 0$, hence its Legendre transform  $G_t$ is well-defined and satisfies
\begin{equation*}
G_t(y) := \sup_{x \geq 0} \{ xy - F_t(x) \} \leq y \left( \frac{ \log (1+y)}{t} \right)^d \mbox{ for all } y \geq0.
\end{equation*}
Young's equality yields $yx \leq F_t(x) + G_t(y)$ and with $ b := \sup \mathrm{supp}(f)$ we find
\begin{equation*}
\int f(\lambda) \xi(\lambda) \mathrm{d}\lambda \leq \int_{- \infty}^b F_t(|\xi(\lambda|) \mathrm{d}\lambda + \int G_t(|f(\lambda)|) \mathrm{d}\lambda.
\end{equation*}
Using part (i), we find that the first integral is bounded by $K_1 e^b$. For the second summand we estimate
\begin{equation*}
\int G(|f(\lambda)|) \mathrm{d}\lambda \leq \int |f(\lambda)| \left( \frac{ \log( 1 + | f(\lambda)|)}{t} \right)^d \mathrm{d}\lambda \leq t^{-d} | \log(1 + \|f\|_\infty)^d \|f\|_1.
\end{equation*}
\end{proof}

%%%%%%%%
%%%%%%%%
\section{Proof of the Wegner estimate}
We will occationally write $H_L(\omega)=H_{\omega, L}$ and $V_L(\omega)=V_{\omega, L}$ for notational convenience. 
Note that for all $\omega \in [\omega_{-}, \omega_{+}]^{\ZZ^d}$, 
all $L \in \NN_{\mathrm{odd}}$ and all $j \in \ZZ^d$, $H_{\omega, L}$ has purely discrete spectrum.
Denote the eigenvalues of $H_{\omega, L}$ by $\left\{ \lambda_i(\omega) \right\}_{i \in \NN}$, enumerated increasingly and counting multiplicities.

	\begin{lemma}
	Let \eqref{eq:abstractUCP} hold and assume that $\omega \in [\omega_{-}, \omega_{+}]^{\ZZ^d}$. 
	Then for all $i \in \NN$ with $\lambda_i(\omega) \in (- \infty , b-1]$ and all $\delta \leq 1/2 - \omega_{+}$ we have
	\begin{equation*}
	\lambda_i(\omega + \delta) \geq \lambda_i(\omega) + \kappa (\delta/2)^{1/M}.
	\end{equation*}
	\end{lemma}

\begin{proof}
		Let $H_{\omega+\delta,L} \phi_i = \lambda_i(\omega+\delta) \phi_i$ for all $i \in \NN$. 
		Then
		\begin{equation*}
		\lambda_i(\omega + \delta) 
		 = \left\langle \phi_i, H_{\omega+\delta,L}\phi_i \right\rangle .
		\end{equation*}
		We write $H_{\omega+\delta, L} = H_{\omega, L} + \left( V_{\omega + \delta,L} - V_{\omega,L} \right)$.  
Note that $\omega_{+} + \delta \leq 1/2 $ and 
\begin{equation}
\label{eq:lower-bound-V}
V_{\omega + \delta,L} (x) - V_{\omega,L} (x) = \sum_{j} \chi_{B{\omega_j + \delta}}(x-j) - \chi_{B\omega_j}(x-j)  
\end{equation}
In the case \eqref{Breather model balls} where the single-site potential is the characteristic function of a ball, this is a 
sum of characteristic functions of mutually disjoint annuli of width $\delta$, cf.~Fig.~1. 
%%%%%%%%
%%%%%%%%
\begin{figure}
  \begin{tikzpicture}
    % draw grid
    \draw (-0.1, -0.1) grid (5 + 0.1, 5 + 0.1);
    %annulus with width 0.2 and random radius in every elementary cell
    \foreach \x in {0,...,4}{
	    \foreach \y in {0, ..., 4}{
		    \pgfmathsetmacro{\w}{0.15*(rand+1)};
			    \filldraw[thick, pattern=north east lines,  even odd rule] (\x + 0.5, \y + 0.5 ) circle (\w + 0.2) -- (\x + 0.5, \y + 0.5) circle (\w);	
    }}		
  \end{tikzpicture}
  \caption{Support of $V_{\omega + \delta} - V_{\omega}$ in the case \eqref{Breather model balls} of characteristic functions of balls as single-site potentials.}
\end{figure}
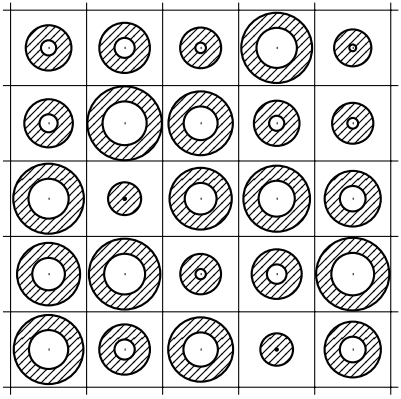
%%%%%%%%
%%%%%%%%
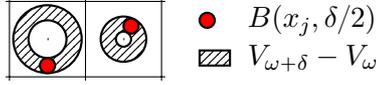
\begin{figure}
 \begin{tikzpicture}[scale = 1]
	\begin{scope}[xshift=2.5cm,yshift=-1.5cm] 
	    % legende

	    \filldraw[thick, fill = red] (0.1, 2.25) circle (0.1);
	    \draw (0.5, 2.25) node[right] {$B(x_j, \delta/2)$};

	    \filldraw[thick, pattern=north east lines] (0,1.65) rectangle (0.4,1.85);
	    \draw (0.5, 1.75) node[right] {$V_{\omega + \delta} - V_{\omega}$};
	\end{scope}
  % draw grid
  \draw (-0.05, -0.05) grid (2.05, 1.05);
  % draw annuli
  \filldraw[thick, pattern=north east lines,  even odd rule] (0.5, 0.5) circle (0.45) -- (0.5, 0.5) circle (0.25);
  \filldraw[thick, pattern=north east lines,  even odd rule] (1.5, 0.5) circle (0.3) -- (1.5, 0.5) circle (0.1);
  % draw small \delta-balls
  \filldraw[thick, fill = red] (0.5, 0.15) circle (0.1);
  \filldraw[thick, fill = red] (1.60, 0.67320508) circle (0.1);
  \end{tikzpicture}
  \caption{Illustration of the increments $V_{\omega + \delta} - V_\omega$ and the choice of the balls $B(x_j, \delta/2)$ in case \eqref{Breather model balls}.}
\end{figure}
%%%%%%%%
%%%%%%%%
Each of these annuli contains a ball of radius $\delta/2$, cf.~Fig.~2, and we find
\begin{equation*}
V_{\omega + \delta,L} (x) - V_{\omega,L} (x) \geq W_{\delta/2,L}(x)
\end{equation*}
where $W_{\delta/2,L}$ is the characteristic function of a union of balls of radius $\delta/2$ each of which is contained in a different elementary cell of the grid $\ZZ^d$.
In the case \eqref{Breather model cubes} where the single-site potential is the characteristic function of a cube, an analogous argument holds.
\\
Restrict now to $i$ such that $\lambda_i^L(\omega) \in(- \infty , b-1]$ which implies in particular $\lambda_i^L(\omega + \delta) \in(- \infty , b]$. 
		Then estimate \eqref{eq:abstractUCP} holds and we obtain via the variational characterization of eigenvalues
		\begin{align*}
		\lambda_i(\omega + \delta) 
		& = \left\langle \phi_i, H_{\omega+\delta,L}\phi_i \right\rangle\\
		& = \max_{\phi \in \mathrm{Span}\{\phi_1, ..., \phi_i\}, \| \phi \| = 1} \left\langle \phi, H_{\omega,L}\phi \right\rangle + \left\langle \phi , \left( V_ {\omega + \delta,L} - V_{\omega,L} \right) \phi \right\rangle\\
		& \geq \max_{\phi \in \mathrm{Span}\{\phi_1, ..., \phi_i\}, \| \phi \| = 1} \left\langle \phi, H_{\omega,L}\phi \right\rangle + \left\langle \phi , W_{\delta/2,L} \phi \right\rangle\\
		& \geq \max_{\phi \in \mathrm{Span}\{\phi_1, ..., \phi_i\}, \| \phi \| = 1} \left\langle \phi, H_{\omega,L}\phi \right\rangle + \kappa (\delta/2)^{1/M}\\
		& \geq \inf_{\mathrm{dim} \mathcal{D} = i} \max_{\phi \in \mathcal{D}, \| \phi \| = 1} \left\langle \phi, H_{\omega,L}\phi \right\rangle + \kappa (\delta/2)^{1/M}\\
		& = \lambda_i(\omega) + \kappa (\delta/2)^{1/M}.
		\end{align*}
\end{proof}
We choose $\delta := 2 \left( \frac{ 4 \epsilon}{\kappa} \right)^M$, so that the lemma becomes
\begin{equation*}
\lambda_i(\omega + \delta) \geq \lambda_i(\omega) + 4 \epsilon.
\end{equation*}
Since we required $\delta \leq 1/2 - \omega_{+}$, this yields the upper bound on $\epsilon$ from the theorem. Note in particular that this bound is smaller than $1/2$.
Now we follow the strategy used in \cite{HundertmarkKNSV-06}.
Let $\rho \in C^\infty(\mathbb{R},[-1,0])$ be a smooth, non-decreasing function such that $\rho = -1$ on $(-\infty; -\epsilon]$ and $\rho = 0$ on $[\epsilon; \infty)$. We can assume $\| \rho' \|_\infty \leq 1/\epsilon$. It holds that
\begin{equation*}
\chi_{[E-\epsilon; E + \epsilon]} (x) \leq \rho(x-E + 2\epsilon) - \rho(x-E-2\epsilon)
\end{equation*}
for all $x \in \mathbb{R}$ which translates into
\begin{align*}
& \mathbb{E} \left[ \mathrm{Tr} \left[ \chi_{[E-\epsilon; E+\epsilon]} ( H_{\omega,L}) \right] \right]\\
\leq \ & \mathbb{E} \left[ \mathrm{Tr} \left[ \rho(H_{\omega,L} - E + 2 \epsilon) - \rho(H_{\omega, L} - E - 2\epsilon) \right] \right]\\
= \ & \mathbb{E} \left[ \mathrm{Tr} \left[ \rho(H_{\omega,L} - E - 2 \epsilon + 4 \epsilon) - \rho(H_{\omega,L} - E - 2\epsilon) \right] \right].
\end{align*}
\begin{lemma}
\begin{equation}\label{trace rho}
\mathrm{Tr} \left[ \rho \left( H_{\omega, L} - E - 2 \epsilon + 4 \epsilon \right) \right] \leq 
\mathrm{Tr} \left[ \rho \left( H_{\omega + \delta, L} - E - 2 \epsilon \right) \right].
\end{equation}
\end{lemma}
	\begin{proof}
		$\rho$ is a monotonous function, hence we have by the previous lemma
		\begin{equation*}
		\rho(\lambda_i(\omega) - E - 2 \epsilon + 4 \epsilon) \leq \rho(\lambda_i(\omega + \delta)- E - 2 \epsilon  ).
		\end{equation*}
		We expand the trace in eigenvalues
		\begin{align*}
		&\mathrm{Tr} \left[ \rho(H_{\omega,L} - E - 2 \epsilon + 4 \epsilon) \right] 
		 = \sum_k \rho(\lambda_i(\omega) - E - 2 \epsilon + 4 \epsilon)  \\
		 \leq &\sum_i \rho(\lambda_i(\omega + \delta) - E - 2 \epsilon)
		 = \mathrm{Tr} \left[ \rho(H_{\omega + \delta, L} - E - 2 \epsilon) \right].
		\end{align*}
	\end{proof}
Now let $\tilde{\Lambda_L} := \Lambda_L \cap \ZZ^d$ and $N := |\tilde{\Lambda}|$. The indices which affect the potential in $\Lambda_L$ will be enumerated by 
\begin{equation*}
k : \left\{ 1, ... N \right\} \rightarrow \tilde{\Lambda_L}, \quad n \mapsto k(n).
\end{equation*}
We define functions which describe how the upper bound in \eqref{trace rho} varies when we change one random variable $\omega_{k(n)}$ while keeping all the other random variables fixed. In order to do that we need some notation. Given 
$\omega \in [\omega_{-}, \omega_{+}]^{\ZZ^d}$, 
$n \in \left\{ 1, ..., N \right\}$, 
$\delta \in [0, 1/2 - \omega_{+}]$ and 
$t \in [\omega_{-},  \omega_{+}]$, 
we define $\tilde{\omega}^{(n, \delta)}(t) \in [\omega_{-}, 1/2]^{\ZZ^d}$ inductively via 
\begin{align*}
\left( \tilde{\omega}^{(1, \delta)}(t) \right)_j &:=
	\begin{cases}
		 t & \mbox{ if } j = k(1)\\
		\omega_j & \mbox{else}
	\end{cases}\\
\left( \tilde{\omega}^{(n, \delta)}(t) \right)_j &:= 
	\begin{cases}
		 t & \mbox{ if } j = k(n)\\
		\left( \tilde{\omega}^{(n-1, \delta)}(\omega_j + \delta) \right)_j & \mbox{else}.
	\end{cases}
\end{align*}
The function 
$\tilde{\omega}^{(n,\delta)} : [\omega_{-},  1/2] \rightarrow [\omega_{-},  1/2]^{\ZZ^d}$ 
is the rank-one perturbation of $\omega$ in the $k(n)$-th coordinate with the additional requirement that all sites $k(1), ..., k(n-1)$ have already been blown up by $\delta$.

We define
\begin{align*}
\Theta_n(t) &:= \mathrm{Tr} \left[ \rho \left( H_L \left({\tilde{\omega}^{(n, \delta)}(t)}\right) - E - 2 \epsilon \right) \right], 
\mbox{ for } n = 1, ..., N.
\end{align*}
Note that 
\begin{equation*}
\Theta_n(\omega_{k(n)}) = \Theta_{n-1}(\omega_{k(n-1)} + \delta) \mbox{ for } n = 2, ..., N
\end{equation*}
and
\begin{align*}
\Theta_N(\omega_{k(N)} + \delta) & = \mathrm{Tr} \left[ \rho \left( H_{\omega + \delta, L} - E - 2 \epsilon \right) \right] ,\\
\Theta_1(\omega_{k(1)}) & = \mathrm{Tr} \left[ \rho \left( H_{\omega, L} - E - 2 \epsilon \right) \right].
\end{align*}
Hence we can expand the expectation of the upper bound in \eqref{trace rho} in a telescopic sum 
\begin{align*}
& \EE \left[ \mathrm{Tr} \left[ \rho(H_{\omega + \delta,L} - E - 2 \epsilon ) \right]- \mathrm{Tr} \left[ \rho(H_{\omega,L} - E - 2 \epsilon  \right]  \right]\\
= & \EE \left[ \Theta_N(\omega_{k(N)} + \delta) - \Theta_1(\omega_{k(1)}) \right]\\ 
= & \sum_{n = 1}^N \EE \left[ \Theta_n(\omega_{k(n)} + \delta) - \Theta_n(\omega_{k(n)}) \right].
\end{align*}
Since we have a product measure structure, we  can apply Fubini's Theorem 
\begin{align*}
\EE \left[ \Theta_n(\omega_{k(n)} + \delta) - \Theta_n(\omega_{k(n)}) \right]
= \EE \left[ \int_{\omega_{-}}^{\omega_{+}} \Theta_n(\omega_{k(n)} + \delta) -\Theta_n(\omega_{k(n)})  \mathrm{d}\mu(\omega_{k(i)})\right].
\end{align*}
Note that for $t \in [\omega_{-},1/2]$, $\Theta_n$ is non-decreasing and bounded.
In fact, monotonicity follows from the inequality 
\begin{equation*}
V_L \left({\tilde{\omega}^{(n, \delta)}(t_1)} \right) \leq V_L \left( {\tilde{\omega}^{(n, \delta)}(t_2)} \right),
\end{equation*}
whenever $t_1 \leq t_2$ and boundedness is due to the fact that $0$ and $1$ provide upper and lower bounds
\begin{equation*}
0 \leq V_L \left({\tilde{\omega}^{(n, \delta)}(t)}\right) \leq 1.
\end{equation*}
	\begin{lemma}
	Let $- \infty < \omega_{-} < \omega_{+} \leq + \infty$.
	Assume that $\mu$ is a probability distribution with bounded density $\nu_\mu$ and support in the interval $[\omega_{-},\omega_{+}]$ or $[\omega_{-}, \omega_{+})$ if $\omega_{+} = \infty$ and let $\Theta$ be a non-decreasing, bounded function. 
	Then for every $\delta > 0$
	\begin{equation*}
	\int_\RR \left[ \Theta(\lambda + \delta) - \Theta(\lambda) \right] \mathrm{d} \mu(\lambda)  \leq \| \nu_\mu \|_\infty \cdot \delta \left[ \Theta(\omega_{+} + \delta) - \Theta(\omega_{-}) \right].
	\end{equation*}
	\end{lemma}
		\begin{proof}
		We calculate
		\begin{align*}
		& \int_\RR \left[ \Theta(\lambda + \delta) - \Theta (\lambda) \right] \mathrm{d} \mu ( \lambda )\\
		 \leq & \| \nu_\mu \|_\infty \int_{\omega_{-}}^{\omega_{+}} \left[ \Theta(\lambda + \delta) - \Theta (\lambda) \right] \mathrm{d} \lambda 
		=   \| \nu_\mu \|_\infty \left[ \int_{\omega_{-} + \delta}^{\omega_{+}+\delta} \Theta(\lambda) \mathrm{d} \lambda - \int_{\omega_{-}}^{\omega_{+}} \Theta(\lambda) \mathrm{d} \lambda \right]\\
		=&  \| \nu_\mu \|_\infty \left[ \int_{\omega_{+}}^{\omega_{+}+\delta} \Theta(\lambda) \mathrm{d} \lambda - \int_{\omega_{-}}^{\omega_{-}+\delta} \Theta(\lambda) \mathrm{d} \lambda \right] 
		\leq   \| \nu_\mu \|_\infty \cdot \delta \left[ \Theta(\omega_{+}+\delta) - \Theta(\omega_{-}) \right].
		\end{align*}
		\end{proof}
Thus, we find
\begin{equation*}
\int_{\omega_{-}}^{\omega_{+}} \left[ \Theta_i(\omega_{k(i)} +  \delta) - \Theta_i(\omega_{k(i)})  \mathrm{d}\mu(\omega_{k(i)})\right] \leq \| \nu_\mu \|_{\infty} \cdot \delta \left[ \Theta_i(\omega_{+}+  \delta) - \Theta_i(\omega_{-}) \right] 
\end{equation*}
\par
We will use the results from the previous section in the following form:
\begin{proposition}\label{Krein}
	Let $H_0 := H_A + V_0$ be a Schr\"o\-ding\-er operator with a bounded potential $V_0\geq 0$, and  
	let $H_1 := H_0 + V$ for some bounded $V \geq 0$ with support in a cube of side length $1$.	
	Denote the Dirichlet restrictions to $\Lambda = \Lambda_L$ by $H_0^\Lambda$ and $H_1^\Lambda$, respectively.
	There are constants $K_1$, $K_2$ depending only on $d$ 
	such that for any smooth function $g: \RR \rightarrow \RR$ with derivative supported in a compact subset of $(-\infty,b]$
	and the property that $g(H_1) - g(H_0)$ is trace class
	\begin{equation*}
	\mathrm{Tr} \left[ g(H_1) - g(H_0) \right] 
\leq K_1 e^{b} + K_2 \left( \ln(1+\| g^\prime \|_\infty )^d \right) \| g^\prime \|_1
	\end{equation*}
	A possible choice is $K_1 := 2 \cdot 32^d \cdot (d+1)!$ and $K_2 := 32^d$.
	\end{proposition}
The expression $\mathrm{Tr} \left[ g(H_1^\Lambda) - g(H_0^\Lambda) \right]$ is well-defined since $H_0^\Lambda$ and $H_1^\Lambda$ are both lower semibounded operators with purely discrete spectrum and only the finite set of eigenvalues in $\mathrm{supp}\ g^\prime$ can contribute to the trace.
\par
Proposition~\ref{Krein} implies
	\begin{lemma}\label{Corollary SSF}
	Let $0 < \epsilon \leq \frac{1}{2}$.
	There is a constant $\tilde{C}$ depending only on $d$ and on $b$ such that 
	\begin{equation*}
	\Theta_n(\omega_{+} + \delta) - \Theta_n(\omega_{-}) \leq \tilde{C} | \ln \epsilon |^d.
	\end{equation*}
	The constant $\tilde{C}$ can be chosen equal to $K_1 e^{b} + 2^d K_2$ with $K_1, K_2$ as in Proposition \ref{Krein}.
	\end{lemma}
		\begin{proof}
		Let $g(\cdot) :=\rho_{E+2\epsilon}( \cdot) := \rho(\cdot - (E+2\epsilon))$.
		By our choice of $\rho$, $g$ has support in $(- \infty, b]$, $\|g^\prime\|_\infty \leq 1/\epsilon$ and $\| g^\prime \|_1 = 1$.
		We define the operators 
		\begin{align*}
		H_0 & := H \left( {\tilde{\omega}^{(n,\delta)}(\omega_{-})} \right) \\
		H_1 & := H \left( {\tilde{\omega}^{(n,\delta)}(\omega_{+} + \delta)} \right).
		\end{align*}
		These are lower semibounded operators with purely discrete spectrum and since $g$ has support in $(- \infty, b]$, the difference $g(H_1)-g(H_0)$ trace class.
		By the previous Proposition
		\begin{equation*}
		\Theta_n(\omega_{+} + \delta) - \Theta_n(\omega_{-}) = \mathrm{Tr} \left[ \rho_{E+2\epsilon}(H_1) - \rho_{E + 2 \epsilon}(H_0) \right] \leq K_1 e^b + K_2 \left( \ln(1 + 1/\epsilon) \right)^d.
		\end{equation*}
		We assumed $0 < \epsilon \leq \frac{1}{2}$, thus $1 + \epsilon \leq \epsilon^ {-1}$ and
		\begin{equation*}
		\ln(1 + 1/\epsilon) =  \ln(1 + \epsilon) - \ln \epsilon  \leq - 2 \ln \epsilon = 2 \lvert \ln \epsilon \rvert
		\end{equation*}
		and $1 \leq \lvert \ln \epsilon \rvert \leq \lvert \ln \epsilon \rvert^d$ which proves the Lemma.
		\end{proof}
Putting everything together yields
\begin{equation}
\EE \left[ \mathrm{Tr} \left[ \chi_{[E- \epsilon, E + \epsilon]}(H_{\omega,L}) \right] \right]
\leq 
\left( K_1 e^{b} + 2^d K_2  \right)
\|\nu_\mu\|_\infty
\cdot
\delta
\left\lvert\ln \epsilon \right\rvert^d L^d.
\end{equation}
and bearing in mind that 
$\delta = 2 \cdot \left(\frac{ 4 \epsilon }{\kappa }\right)^{1/M}$
we obtain \eqref{eq: Wegner Breather}.
This proves Theorem \ref{Wegner Breather}.

%%%%%%%%%%%%%%%%%%
%     \bibliographystyle{abbrv}
%     \bibliography{ILit-15}

\begin{thebibliography}{10}

\bibitem{BroderixHL-00}
K.~Broderix, D.~Hundertmark, and H.~Leschke.
\newblock Continuity properties of {S}chr\"odinger semigroups with magnetic
  fields.
\newblock {\em Rev. Math. Phys.}, 12(2):181--225, 2000.

\bibitem{CombesHM-96}
J.-M. Combes, P.~D. Hislop, and E.~Mourre.
\newblock Spectral averaging, perturbation of singular spectra, and
  localization.
\newblock {\em Trans. Amer. Math. Soc.}, 348(12):4883--4894, 1996.

\bibitem{CombesHN-01}
J.-M. Combes, P.~D. Hislop, and S.~Nakamura.
\newblock The ${L}^p$-theory of the spectral shift function, the {Wegner}
  estimate, and the integrated density of states for some random
  {Schr\"odinger} operators.
\newblock {\em Commun. Math. Phys.}, 70(218):113--130, 2001.

\bibitem{Evans-98}
L.~C. Evans.
\newblock {\em Partial differential equations}, volume~19 of {\em Graduate
  Studies in Mathematics}.
\newblock American Mathematical Society, Providence, RI, 1998.

\bibitem{HundertmarkKNSV-06}
D.~Hundertmark, R.~Killip, S.~Nakamura, P.~Stollmann, and I.~Veseli\'c.
\newblock Bounds on the spectral shift function and the density of states.
\newblock {\em Comm. Math. Phys.}, 262(2):489--503, 2006.
\newblock ArXiv.org/math-ph/0412078.

\bibitem{HundertmarkS-02}
D.~Hundertmark and B.~Simon.
\newblock An optimal ${L}^p$-bound on the {Krein} spectral shift function.
\newblock {\em J. Anal. Math.}, 87:199--208, 2002.
\newblock http://www.ma.utexas.edu/mp\_arc/c/00/00-370.ps.gz.

\bibitem{HundertmarkS-04}
D.~Hundertmark and B.~Simon.
\newblock A diamagnetic inequality for semigroup differences
\newblock {\em J. Reine Angew. Math.}, 571:107--130, 2004

\bibitem{KirschV-10}
W.~Kirsch and I.~Veseli\'c.
\newblock Lifshitz tails for a class of Schr\"odinger operators with random
  breather-type potential.
\newblock {\em Lett. Math. Phys.}, 94(1):27--39, 2010.

\bibitem{Klein-13}
A.~Klein.
\newblock Unique continuation principle for spectral projections of
  {Schr{\"o}dinger} operators and optimal {Wegner} estimates for non-ergodic
  random schr{\"o}dinger operators.
\newblock {\em Comm. Math. Phys.}, 323(3):1229--1246, 2013.

\bibitem{KloppLNS-12}
F.~{Klopp}, M.~{Loss}, S.~{Nakamura}, and G.~{Stolz}.
\newblock Localization for the random displacement model.
\newblock {\em Duke Math. J.}, 161(4):587--621, 2012.

\bibitem{Leinfelder-83}
H.-~{Leinfelder}.
\newblock Gauge invariance of {S}chr\"odinger operators and related spectral properties.
\newblock {\em J. Operator Theory}, 9:163--179, 1983.

\bibitem{NakicTTV-15}
I.~{Naki{\'c}}, M.~{T{\"a}ufer}, M.~{Tautenhahn}, and I.~{Veseli{\'c}}.
\newblock {Scale-free uncertainty principles and applications}.
\newblock Working paper, Technische Universit\"at Chemnitz.

\bibitem{NakicTTV-14}
I.~{Naki{\'c}}, M.~{T{\"a}ufer}, M.~{Tautenhahn}, and I.~{Veseli{\'c}}.
\newblock {Scale-free uncertainty principles and Wegner estimates for random
  breather potentials}.
\newblock ArXiv.org/1410.5273.
\newblock To appear in \emph{Comptes Rendus Mathematique}.

\bibitem{RojasMolinaV-13}
C.~Rojas-Molina and I.~Veseli\'c.
\newblock Scale-free unique continuation estimates and applications to random
  {Schr\"odinger} operators.
\newblock {\em Comm. Math. Phys.}, 320(1):245--274, 2013.
\newblock http://ArXiv.org/1110.4652.

\bibitem{SchumacherV}
C.~Schumacher and I.~Veseli\'c.
\newblock Lifschitz tails for random breather potentials.
\newblock Working paper, Technische Universit\"at Chemnitz.

\bibitem{Simon-79c}
B.~Simon.
\newblock Functional integration and quantum physics
\newblock {\em Pure and Applied Mathematics}, 86, Academic Press Inc. [Harcourt Brace Jovanovich Publishers], New York, 1979.

\bibitem{Veselic-07}
I.~Veseli\'c.
\newblock Lifshitz asymptotics for {Hamiltonians} monotone in the randomness.
\newblock {\em Oberwolfach Rep.}, 4(1):380--382, 2007.
\newblock ArXiv.org/0708.0487.

\bibitem{Veselic-07b}
I.~Veseli\'c.
\newblock {\em {\it Existence and regularity properties of the integrated
  density of states of random {Schr\"odinger} Operators}}, volume Vol. 1917 of
  {\em Lecture Notes in Mathematics}.
\newblock Springer-Verlag, 2008.

\bibitem{Yafaev-92}
D.~Yafaev.
\newblock {\em Mathematical Scattering Theory}.
\newblock Translations of Mathematical Monographs, {\bf 105}. American
  Mathematical Society, Providence, Rhode Island, 1992.
\newblock [Russian original: Izdatel{\cprime} stvo Sankt-Peterburgskogo
  Universiteta, St. Petersburg, 1994].

\end{thebibliography}
%%%%%%%%%%%%%%%%%%

  \def\cprime{$'$}\def\polhk#1{\setbox0=\hbox{#1}{\ooalign{\hidewidth
     \lower1.5ex\hbox{`}\hidewidth\crcr\unhbox0}}}
 \def\polhk#1{\setbox0=\hbox{#1}{\ooalign{\hidewidth
    \lower1.5ex\hbox{`}\hidewidth\crcr\unhbox0}}}

\end{document}